\documentclass[12pt]{article}
\usepackage{amsmath}
\usepackage{amssymb,amsthm}
\usepackage{color}

\newtheorem{Thm}{Theorem}[section]
\newtheorem{theorem}[Thm]{Theorem}
\newtheorem{proposition}[Thm]{Proposition}

\newtheorem{remark}{Remark}[section]

\title{A relation of thermodynamic relevance between the superadditivity, concavity and homogeneity properties of real-valued functions}

\author{Walter F. Wreszinski\footnote{wreszins@gmail.com, 
Instituto de Fisica, Universidade de S\~ao Paulo (USP), Brasil}}        
        
\begin{document}

\maketitle

\begin{abstract}
We provide a necessary and sufficient condition for the validity of the following Landsberg-Thirring theorem: for a real-valued function on a convex set, any two of the properties of superadditivity, concavity and homogeneity implies the third. Applications to statistical thermodynamics, following Thirring and Landsberg, are briefly revisited. 
\end{abstract}

\section{Introduction and summary}

In \cite{L}, Landsberg studied the issue of whether equilibrium is always an entropy maximum, having in view nonextensive (e.g., gravitational) systems, which has been recently revived under a different point of view in \cite{LYEM}. In the process, he arrived at a connection between the properties of homogeneity, superadditivity and concavity of a real-valued function. Unfortunately, he did not formulate this connection as a theorem. Thirring attempted to do so in his beautiful introduction to Lieb's selecta \cite{ThirrLi} as follows (we present the version used by Landsberg, which replaces subadditivity by superadditivity and convexity by concavity): 

\begin{proposition}
\label{prop:T}
Let $x \to f(x)$ be a map from a convex set of $\mathbf{R}^{d}$ into $\mathbf{R}$. Then any two of the conditions
\begin{itemize}
\item [$a.)$] (H) (Homogeneity) $f(\lambda x) = \lambda f(x) \mbox{ for all } \lambda \in \mathbf{R}_{+}$;
\item [$b.)$] (Sp) (Superadditivity) $f(x_{1}+x_{2}) \ge f(x_{1}) + f(x_{2})$;
\item [$c.)$] (Cc) (Concavity) $f[\lambda x_{1}+ (1-\lambda) x_{2}] \ge \lambda f(x_{1}) + (1-\lambda) f(x_{2})$ for $0\le \lambda \le 1$.
\end{itemize}
implies the third.
\end{proposition}

The formulae $a.)$, $b.)$, $c.$ are equivalent to the following \cite{ThirrLi}:
\begin{itemize}
\item [$a.)$] (H) (Homogeneity) $f(\lambda x) = \lambda f(x) \mbox{ for all } \lambda \in \mathbf{R}_{+}$;
\item [$b.)$] (S) (Subadditivity) $f(x_{1}+x_{2}) \le f(x_{1}) + f(x_{2})$;
\item [$c.)$] (Cv) (Convexity) $f[\lambda x_{1}+ (1-\lambda) x_{2}] \le \lambda f(x_{1}) + (1-\lambda) f(x_{2})$ for $0\le \lambda \le 1$.
\end{itemize}

Moreover \cite{ThirrLi} (H) and (S) are conditions for stability agains implosion and explosion, respectively, and (C) is the thermodynamic stability condition.

Proposition ~\ref{prop:T} cannot be true as stated, due to the counterexample given in the proof of the forthcoming theorem ~\ref{th:2.1} in section 2. The additional assumption required there (\eqref{(2.1.2)} of section 2, together with Assumption A) throws some additional light into the proposed relationship between (H), (Sp) and (Cc), because it yields a necessary and sufficient condition for the validity of a theorem of type of theorem ~\ref{prop:T}. Its physical significance will be left to the conclusion in section 4, after the applications to statistical thermodynamics, due to Thirring \cite{ThirrLi} and Landsberg \cite{L}, have been briefly revisited in section 3.

The main ideas in the proof of the forthcoming theorem, whose statement and proof are provided in pages 3-6, are due to the late Peter Landsberg and Walter Thirring, and therefore we call it the Landsberg-Thirring theorem. Our own modest contribution was to find (what we believe is) the natural framework for a theorem, which may, however, be of general interest in real analysis, because, on the one hand, it relates three basic properties of real-valued functions, and, on the other, there seem to be few necessary and sufficient criteria for super (sub) additivity, since the classic works (\cite{HP}, \cite{Ros})- see, however, \cite{Bru}. Possible generalizations to a noncommutative setting may also be envisaged \cite{UUG}, for an introduction see also section 8.1 of \cite{Carlen} and references given there. 

\section{Main Theorem}

The function $f_{0}$ in the counterexample to proposition ~\ref{prop:T} given in the proof of theorem ~\ref{th:2.1} below is defined on a convex \emph{open} set $X$ and exhibits a singularity (of the second kind) at a point (chosen without loss of generality as the origin of the Cartesian coordinate system). This motivates the introduction of the following simple framework. 

\emph{Assumption A} Let $X$ be a convex open cone in $\mathbf{R}^{d}$, $1\le d <\infty$, and $\overline{X}$ denote its closure in $\mathbf{R}^{d}$. Thus, $\overline{X}$ is a closed convex cone (see, e.g., \cite{Fenchel}), and 
\begin{equation} 
(0, \cdots, 0) \in \overline{X} \setminus X
\label{(2.1.1)}
\end{equation}
By definition, $X$ is closed under addition and multiplication by a scalar in $\mathbf{R}_{+}$.
\space
\begin{theorem}
\label{th:2.1}
Let $X$ be as in Assumption A, and $f$ be a real-valued function on $X$.
A necessary and sufficient condition for the statement that any two of the properties (H), (Sp) and (Cc) for $f$ imply the third is
\begin{equation}
\liminf_{(x_{1},\cdots,x_{d}) \to (0,\cdots,0)} f(x_{1}, \cdots, x_{d}) \ge 0
\label{(2.1.2)}
\end{equation}
\end{theorem}

\begin{proof}

Let Assumption A be valid. We need only show that 
\begin{equation}
(Cc) \land (Sp) \Rightarrow (H)
\label{(2.1.3)}
\end{equation}
if and only if \eqref{(2.1.2)} holds. We first show necessity. Let $d=1$, $X= \mathbf{R}_{+}=(0,\infty)$ and
\begin{equation}
0<c< \infty
\label{(2.1.4)}
\end{equation}
be given and define
\begin{equation}
\label{(2.1.5)}
h(x)=\begin{cases}\log(cx),& \mbox{ if } 0<x\le\frac{2}{c}\\0, \mbox{ otherwise } \end{cases}
\end{equation}
and
\begin{equation}
\label{(2.1.6)}
g(x)=\begin{cases} \log(2)+(\frac{c}{2})^{2}(x-\frac{2}{c}),& \mbox{ if } \frac{2}{c}\le x< \infty\\0, \mbox{ otherwise } \end{cases}
\end{equation}
Notice that we have chosen the angular coefficient of the straight line equal to the tangent to the logarithmic function
at the point $\frac{2}{c}$.
Further, define the function on $(0,\infty)$:
\begin{equation}
\label{(2.1.7)}
f_{0}(x) = h(x) + g(x)
\end{equation}
By \eqref{(2.1.5)} and \eqref{(2.1.6)}, $f_{0}$ is continuous, and has a continuous derivative at the point $\frac{2}{c}$.
The function $h$ is superadditive on $(0,\frac{2}{c}]$ because
\begin{equation}
\label{(2.1.8)}
\log[(c(x_{1}+x_{2})] \ge \log(cx_{1})+\log(cx_{2})=\log(c^{2}x_{1}x_{2})
\end{equation}
is true whenever 
$$
x_{1}+x_{2} \ge cx_{1}x_{2}
$$
or
\begin{equation}
\label{(2.1.9)}
\frac{x_{1}+x_{2}}{x_{1}x_{2}} \ge c
\end{equation}
Superadditivity of $h$ on $(0,\frac{2}{c}]$ means, by definition \eqref{(2.1.5)}, that \eqref{(2.1.8)} holds 
for all $x_{1},x_{2} \in (0,\frac{2}{c}]$ such that $x_{1}+x_{2}$ is also an element of $(0,\frac{2}{c}]$, i.e., such that
\begin{equation}
\label{(2.1.10)}
0<x_{1}+x_{2}\le \frac{2}{c}
\end{equation}
By \eqref{(2.1.9)}, \eqref{(2.1.8)} holds under \eqref{(2.1.10)} due to the elementary inequalities
\begin{equation}
\label{(2.1.11)}
\frac{1}{x_{1}} + \frac{1}{x_{2}} \ge \frac{2}{x_{1}^{1/2}x_{2}^{1/2}} \ge c
\end{equation}
We now consider the remaining case 
\begin{equation}
\label{(2.1.12)}
x_{1}+x_{2} > \frac{2}{c}
\end{equation}
In case \eqref{(2.1.12)}, we may have two different cases:
\begin{itemize}
\item [$a.)$] $x_{1} \le \frac{2}{c}$ and $x_{2}>\frac{2}{c}$;
\item [$b.)$] $x_{1}>\frac{2}{c}$ and $x_{2}> \frac{2}{c}$
\end{itemize}
Of course, the case a.) with $x_{1}$ and $x_{2}$ exchanged is the same. 
In case a.),
$$
f_{0}(x_{1}+x_{2}) = g(x_{1}+x_{2}) = \tilde{g}(x_{1})+g(x_{2}) \ge h(x_{1})+g(x_{2})=f_{0}(x_{1})+f_{0}(x_{2})
$$
where $\tilde{g}$ denotes the natural extension of $g$ to $(0,\infty)$, by the remark after equation \eqref{(2.1.6)}.
In case b.), 
$$
f_{0}(x_{1}+x_{2}) = g(x_{1}+x_{2}) = g(x_{1})+g(x_{2})=f_{0}(x_{1})+f_{0}(x_{2})
$$
which completes the proof of superadditivity of the function $f_{0}$.

The function $h$, defined by \eqref{(2.1.5)}, satisfies $\frac{d^{2}h}{dx^{2}} \le 0$
under condition \eqref{(2.1.4)}, and is, therefore, concave on $(0,\frac{2}{c})$, and $g$, defined by \eqref{(2.1.6)},
being linear, is concave as well on $(\frac{2}{c},\infty)$. The function $f_{0}$, given by \eqref{(2.1.7)}, is continuous
on $0,\infty)$ and has the property that through every point of the curve $y=f_{0}(x)$ there is at least one line
which lies wholly above the curve. Indeed, for the point $x=\frac{2}{c}$, at which the second derivative of $f_{0}$ is
discontinuous, such a line is the tangent to the curve at the point. Thus, by \cite{HLP}, p. 95, $f_{0}$ is concave on
$(0,\infty)$. This example trivially generalizes to $\mathbf{R}^{d}$, by taking
\begin{equation}
\label{(2.1.13)}
\tilde{f}_{0}(x_{1}, \cdots x_{d}) = \sum_{i=1}^{d} f_{0}(x_{i})
\end{equation}
for $(x_{1}, \cdots x_{d}) \in \mathbf{R}_{+} \times \cdots \times \mathbf{R}_{+}$.
Finally, (H) obviously fails for $f_{0}$, and consequently for $\tilde{f}_{0}$, and necessity is proved.

In order to show sufficiency, assume a real-valued function $f$ satisfies \eqref{(2.1.2)} and both (Cc) and (Sp) on a convex open $X \in \mathbf{R}^{d}$ satisfying Assumption A. By (Cc),
\begin{equation}
f(\lambda x + (1-\lambda)y) \ge \lambda f(x) + (1-\lambda) f(y) \mbox{ with } x,y \in X \mbox{ and } 0 \le \lambda \le 1
\label{(2.2)}
\end{equation}   
Since $f$ satisfies \eqref{(2.2)} on an open set, it is continuous there (see \cite{HLP}, Theorem 111, p.91), and therefore \eqref{(2.2)} yields, for all $x \in X$,
\begin{eqnarray*}
\lim_{y \to (0,\cdots 0)} f(\lambda x+ (1-\lambda)y) = f(\lambda x) \ge\\
\ge \lambda f(x) + (1-\lambda) \liminf_{y \to (0,\cdots 0)} f(y) 
\end{eqnarray*}
from which, by \eqref{(2.1.2)},
\begin{equation}
\label{(2.3)}
f(\lambda x) \ge \lambda f(x)
\end{equation}
Choosing, now, $n \in \mathbf{N}$ and $\lambda = \frac{1}{n}$, we obtain from \eqref{(2.3)}
\begin{equation}
\label{(2.4)}
f(x) \le n f(\frac{x}{n}) \mbox{ for all } x \in X
\end{equation}
We further obtain from (Sp), 
\begin{equation}
\label{(2.5)}
f(nx) \ge n f(x) \mbox{ for all } x \in X
\end{equation}
Let $w \in X$ and, given $n \in \mathbf{N}$, define $x \in X$ by $nx=w$. Then, by \eqref{(2.5)},
\begin{equation}
\label{(2.6)}
f(w) \ge n f(\frac{w}{n})
\end{equation}
From \eqref{(2.4)} and \eqref{(2.6)},
\begin{equation}
\label{(2.7)}
f(w) = n f(\frac{w}{n}) \mbox{ for all } w \in X
\end{equation}
By \eqref{(2.7)}, replacing $n$ by $m$, and writing $u=\frac{w}{m}$, we find $f(mu)=mf(u) \mbox{ for all } u \in X$, and, finally,
\begin{equation}
\label{(2.8)}
f(n^{-1}mu) = n^{-1}m f(u) \mbox{ for all } u \in X \mbox{ and for all } n,m \in \mathbf{N}
\end{equation}
Take, now, any irrational number $\lambda \in \mathbf{R}_{+}$, and let $\frac{p_{k}}{q_{k}}$ be the continued fraction approximants 
of $\lambda$ (\cite{Khinchin}, p.18). By the continuity of $f$, $f(\frac{p_{k}}{q_{k}}u) \to f(\lambda u)$ as $k \to \infty$
and \eqref{(2.8)} finally yields (H).

\end{proof}

In the case $d=1$, that is, for functions $f$ of one real variable, the l.h.s. of \eqref{(2.1.3)} implies that the function $f$ 
which satisfies (H) is in fact trivial:

\begin{proposition}
\label{prop:3.1}
If $d=1$, under assumption \eqref{(2.1.2)}, \eqref{(2.1.3)} implies that the function satisfying (H) is trivial,
i.e., $f=cx$ for $c$ a given constant. The analogue of this assertion no longer holds if $d=2$.
\end{proposition}

\begin{proof}
Let $h(x)=\frac{f(x)}{x}$. By \cite{HLP}, Theorem 103, p.83, and \cite{HP}, p.239, under the assumed
concavity, $f$ is superadditive on $(0,\infty)$ iff $h$ is nondecreasing. Thus, the l.h.s. of \eqref{(2.1.3)}
implies that $h$ is nondecreasing. Let $x \in (0,\infty)$, $\lambda \in (0,1]$
be given. Then, by \eqref{(2.3)},
$$
f(\lambda x) \ge \lambda f(x)
$$
by the assumption \eqref{(2.1.2)}. Division by $\lambda x$ gives $\frac{f(\lambda x)}{\lambda x} \ge \frac{f(x)}{x}$. Thus $h$ is nonincreasing
on $(0,\infty)$, and, by the previous result, it must also be nondecreasing and thus be a constant 
on $(0,\infty)$, i.e., $\frac{f(x)}{x}=c$. 

For $d=2$ the example $s_{ph}(E,V)$ in subsection 3.3 provides a counterexample to the assertion of the proposition, since
$s_{ph}(E,V) \ne c_{1}E+c_{2}V$, with $c_{1},c_{2}$ given constants.

\end{proof}

\begin{remark}
\label{Remark 2.1}
If $\phi$ is a function on $[0,1]$ defined by $\phi(x)=0$ if $x \in [0,1)$, while $\phi(1)=1$, it is convex on $[0,1]$ but not continuous on $[0,1]$. Since continuity is important in in the proof of theorem ~\ref{th:2.1}, this is an indication that the assumption that the domain of the function $f$ is an open convex set $X$ in theorem ~\ref{th:2.1} is natural even in the case that the function is bounded in the closure $\overline{X}$ of $X$.
\end{remark}

\begin{remark}
\label{Remark 2.2}
The counterexample \eqref{(2.1.13)} is suggested by the entropy of a free, classical gas.
\end{remark} 

\begin{remark}
\label{Remark 2.3}
Proposition ~\ref{prop:3.1} shows that theorem ~\ref{th:2.1} is non-trivial only in the case $d \ge 2$, i.e., in the case
of functions of several variables. In the applications to statistical thermodynamics, to which we now turn, one is
typically concerned with at least two variables, as the forthcoming three examples demonstrate.
\end{remark}

\section{Applications to statistical thermodynamics}

In this section we briefly review three applications to statistical thermodynamics, on the light of theorem ~\ref{th:2.1}. In section 3.1 we revisit non-relativistic gravitational systems, following Thirring \cite{ThirrLi}. This is the most important application, in which theorem ~\ref{th:2.1} is natural, because the origin lies outside the range of physical values of the variables involved. The second one, in section 3.2, is the Kerr-Newman black-hole, and is due to Landsberg \cite{L}, but there no thermodynamic limit is involved, similarly to the third one, free photons, in section 3.3.

\subsection{Non-relativistic gravitational systems}

This application is based on the model of $N$ fermions interacting via Newtonian attractions, as represented by \eqref{(3.1.1)}. The quantum thermodynamics of this model was derived by Hertel, Narnhofer and Thirring \cite{HNT}.

In \cite{HNT}, the system of $N$ electrically neutral, massive fermions of one species, interacting by Newtonian forces, was studied as a model of a neutron star. We shall follow the excellent summary by Sewell \cite{Se} (see also \cite{NaSe} for a detailed study of the equilibrium states of the system). The Hamiltonian on a Hilbert space ${\cal H}_{N,V}$ is given by
\begin{equation}
H_{N,V} = -\frac{\hbar^{2}}{2m}\sum_{j=1}^{N} \Delta_{j}-\kappa m^{2}\sum_{j\ne k \in V;j,k=1}^{N}r_{jk}^{-1}
\label{(3.1.1)}
\end{equation}
where $\kappa$ is the gravitational constant, $\Delta_{j}$ the Laplacian for the $j-th$ particle, and $r_{j,k}$ the distance between the $jth$ and the $kth$ particle. Due to the long-range and attractive nature of the interactions, it is well explained in \cite{Se} and derived in \cite{HNT} that the statistical thermodynamics of the model may be formulated within a framework in which, for each value of $N$, the system is confined to a spatial region $\Omega_{N}$ such that the volume $V_{N}$ of $\Omega_{N}$ is proportional to $N^{-1}$. It may then be proved that if the energy $E_{N}$ of the system is such that $N^{-7/3}E_{N} \to e$ and $NV_{N} \to v$ as $N \to \infty$, the microcanonical specific entropy $N^{-1}S_{N}$ converges to a function $s$ of $(e,v)$. In order to preserve the original \emph{extensive} variables $E,V,N$, we proceed as in (\cite{ThirrLi}, p.5), defining the function
\begin{equation}
s(E,V,N) \equiv \lim_{\lambda \to \infty} \frac{1}{\lambda} S(\lambda^{-7/3}E, \lambda^{-1}V, \lambda N)
\label{(3.1.2)}
\end{equation}
where
\begin{equation}
S(E,V,N) \equiv \log \dim ({\cal H}_{N,V}^{E})
\label{(3.1.3)}
\end{equation}
and ${\cal H}_{N,V}^{E}$ denotes the subspace of ${\cal H}_{N,V}$ satisfying the condition
\begin{equation}
Tr_{{\cal H}_{N,V}^{E}} H(N,V) < E
\label{(3.1.4)}
\end{equation}  
We shall in the following consider a \emph{fixed} number $N$ of particles (as in \cite{LYPR}), and take the limit in \eqref{(3.1.2)} along $\lambda \in \mathbf{N}$, i.e., along the positive integers: we shall denote the function so defined by $s(E,V)$. This number $N$ is assumed to be an integer, and $N \ge 1$. Note that $N=0$ is a \emph{crucial} value in \cite{L}, but would be ambiguous in \eqref{(3.1.2)}.

Our first application of theorem ~\ref{th:2.1} consists in choosing there $d=2$ and $X \equiv \mathbf{R}_{-}^{E} \times \mathbf{R}_{+}^{V}$, where the superscripts refer to the variables $E$ and $V$. This choice satisfies Assumption A. Indeed, it follows from the framework just described that, and the attractive character of the interactions in \eqref{(3.1.1)}, that
\begin{equation}
-\infty < E < 0
\label{(3.2)}
\end{equation}
as well as
\begin{equation}
0 < V < \infty
\label{(3.3)}
\end{equation}
are the ranges of the variables $E$ and $V$. Further, both the point $(E,V)=(0,0)$ and the half-axes $(\mathbf{R}_{-},0)$ and $(0,\mathbf{R}_{+})$ lie in the complementary region of the physical values of the quantities $(E,V)$. The (quantum) microcanonical entropy $S(E,V,N)= \log \dim ({\cal H}_{N,V}^{E}) \ge 0$, and thus
\begin{equation}
s(E,V) \ge 0 \mbox{ if } (E,V) \in \mathbf{R}_{-} \times \mathbf{R}_{+}
\label{(3.4)}
\end{equation}
Therefore \eqref{(2.1.2)} holds for $s(E,V)$, i.e.,
\begin{equation}
\liminf_{(E,V) \to (0,0)} s(E,V) \ge 0
\label{(3.5)}
\end{equation}
By construction, $s(E,V)$ does not satisfy (H), but it does satisfy (Sp). The latter property is most easily seen to hold from the property of the inverse function $e(S,V)$, which is \emph{subadditive} (\cite{ThirrLi}, \cite{HNT}), as a consequence of the attractive nature of the interactions (see also \cite{Ru}, pp. 42,65). We thus arrive at

\begin{proposition}
\label{prop:3.1}
The function $s(E,V)$ does not satisfy (Cc), i.e., thermodynamic stability fails for model \eqref{(3.1.1)}.
\end{proposition}
(Cc) is the condition of thermodynamic stability. The standard manifestation of \emph{non-} (Cc) is the nonpositivity of the specific heat: $s$ becomes convex with respect to $E$, leading to a phase transition of van der Waals type. For model \eqref{(3.1.1)} this was shown in \cite{HNT} (see also \cite{HT} for a soluble model). We refer to \cite{Thirr1} for the discussion of the stage in the stellar evolution in which such behavior is expected.

\begin{remark}
\label{Remark 3.1}
When the specific entropy is regarded as a function of the state, represented by a density matrix $\rho$, then $s$ is \emph{subadditive} rather than superadditive (see, e.g., \cite{Carlen}, Theorem 6.5, pg. 122, and further for a comprehensive account of various related results). In this framework, instead of continuity of $s$ , upper semicontinuity obtains instead \cite{Sewell1}. This property is crucial in the dynamic proof of the second law in \cite{therm3}.
\end{remark}

As remarked in \cite{Se}, for very large $N$ (of the order of $10^{60}$), the nonrelativistic model \eqref{(3.1.1)} becomes unphysical, because the mean particle velocities become comparable to the velocity of light. If the star's mass exceeds the Chandrasekhar limit, rigorously analysed in \cite{LYau}, it is believed that its collapse leads to a black-hole. A very nice account of black-hole thermodynamics may be found in section 4 of \cite{Se}, see also \cite{SePLA}. A special model thereof is the Kerr-Newman black-hole, considered in \cite{L}, to which we now come.

\subsection{The Kerr-Newman black-hole}

A (classical) Kerr-newman black-hole of charge $Q$, angular momentum $J$ and mass (energy) $M$ is assumed to be described by the Beckenstein-Hawking entropy $S_{BH}$, defined by (see \cite{Wald}):
\begin{equation}
S_{BH} = \pi(2M^{2}+2M\sqrt{M^{2}-a^{2}-Q^{2}}-Q^{2})
\label{(3.6)}
\end{equation}
where
\begin{equation}
a = \frac{J}{M}
\label{(3.7)}
\end{equation} 
and the inequality
\begin{equation}
M^{2} > a^{2}+Q^{2}
\label{(3.8)}
\end{equation}
is assumed. A rigorous derivation of the second law of thermodynamics for black-holes is given in \cite{SePLA}. As remarked and explained by Sewell (\cite{Se}, \cite{SePLA}), $S_{BH}$ has, at most, an information-theoretic content, not a statistical thermodynamic one. This fact is due to the idealization involved in models such as \eqref{(3.6)}-\eqref{(3.8)}, essentially the characterization of a black-hole by only three quantities, without any microstructure (Wheeler's ``no-hair theorem''). It follows from the explicit formula \eqref{(3.6)} that $S_{BH}$ does not satisfy (H). Let, for simplicity, $Q=0$, and define the variables $x=(M,J)$, with ranges $M \in (0,\infty)$ and $J \in (0,\infty)$. It may be shown that $S_{BH}$ is strictly superadditive, i.e.,
\begin{equation}
S_{BH}(x_{1}+x_{2}) > S_{BH}(x_{1}) + S_{BH}(x_{2})
\label{(3.9)}
\end{equation}
(see \cite{L}, p. 161)). Choosing $X= \mathbf{R}_{+} \times \mathbf{R}_{+}$, we see that  Theorem ~\ref{th:2.1} is (trivially) applicable to $S_{BH}(x)$, with  $S_{BH}$ continuous at $x=(0,0)$, and $S_{BH}(0,0)=0$, asserting that thermodynamic stability (Cc) must be violated. It is straightforward to show this for the Kerr-Newman black-hole, because $S_{BH}$ is twice continuously differentiable in each of its variables, and, in this case, a necessary condition for (Cc) is (see, e.g., \cite{HLP},p.81, (3.12.4))
\begin{equation}
\frac{\partial^{2}S_{BH}}{\partial M^{2}} \le 0
\label{(3.10)}
\end{equation}
We find
\begin{equation}
\frac{\partial^{2}S_{BH}}{\partial M^{2}} \ge 4
\label{(3.11)}
\end{equation}
and, therefore, thermodynamic stability indeed fails for the Kerr-Newman black-hole, as predicted by Theorem ~\ref{th:2.1}. Thus, altough $S_{BH}$ is not a true statistical thermodynamical entropy, by Sewell's previously mentioned remarks, it is, nevertheless, reassuring that it has ``inherited'' the instability properties non-(H) and non-(Cc) from the models of nonrelativistic gravitational systems.

\subsection{The free photon gas}

The entropy of the free photon gas has curious properties from the thermodynamical standpoint (see \cite{LYPR} and references given there). It is given by
\begin{equation}
s_{ph}(E,V) = E^{3/4}V^{1/4}
\label{(3.12)}
\end{equation}
defined on $\mathbf{R}_{+} \times \mathbf{R}_{+}$, with $s_{ph}$ continuous at $(0,0)$, and $s(0,0)=0$, thus satisfying \eqref{(2.1.2)} trivially. $s_{ph}$ clearly satisfies (H). Further, it is easily seen to be concave by (\cite{HLP}, p.81, (3.12.4)), but \emph{not} strictly concave, because, denoting partial derivatives by superscripts, $(s_{ph}^{EV})^{2}-s_{ph}^{EE}s_{ph}^{VV}=0$, and, by (3.12.5) of \cite{HLP}, p.81, $(s^{EV})^{2}-s^{EE}s^{VV}>0$ is necessary for strict concavity of a function $s$. As a consequence of theorem ~\ref{th:2.1}, $s_{ph}$ satisfies (Sp), a property which is (surprisingly) cumbersome to prove directly. Strict superadditivity does not, however, follow from the theorem, because of the previous remark on strict concavity.

\section{Conclusion}

As remarked by Thirring (\cite{ThirrLi}, p. 5), (H) has the interpretation of stability against implosion. On the other hand, the property of subadditivity of the inverse function $e(S,V)$ has the interpretation of stability against explosion: ``one gains energy by putting two parts together''. Due to the Newtonian attraction, therefore, strict superadditivity of $s(E,V)$ is therefore expected, which indeed holds for the model in subsection 3.1, and is even inherited by the black-hole model of section 3.2. This intuition is, of course, not applicable to the free photon gas of subsection 3.3, and, indeed, as remarked there, strict superadditivity is not a direct consequence of theorem ~\ref{th:2.1}. 

Another indication that subadditivity of the energy may be the general property responsible for the equivalence between stability in the sense of (H) and thermodynamic stability is, as remarked by Thirring in \cite{ThirrLi}, the universal character of van der Waals forces (which are attractive) for neutral assemblies of atoms or molecules \cite{LT}.

The main general obstruction to the validity of (2) in applications to statistical thermodynamics (with $f$ taken as the entropy function) is posed by \emph{classical} statistical mechanics: the classical entropy is unbounded from below \cite{RR} (see also remark ~\ref{Remark 2.1}). This adds a further link to the importance of quantum theory in various stability and instability aspects of cosmic bodies.

\section{Acknowledgements}

I am very much indebted to Geoffrey Sewell for pointing out an error in a previous version of this manuscript, and thank him, as well as Pedro L. Ribeiro, for their interest. Special thanks are due to Elliott Lieb for his kind remarks on this paper, as well as a correction. I should like to thank the referee for useful remarks, as well as for his positive appraisal of the present paper.

\end{document}